\documentclass[letterpaper, 10 pt, conference]{ieeeconf}  

\IEEEoverridecommandlockouts
\usepackage{amsthm}
\usepackage{balance}
\newtheorem{rem}{\textbf{Remark}}
\newtheorem{lemma}{\textbf{Lemma}}

\newtheorem{thm}{\textbf{Theorem}}
\newtheorem{defi}{\textbf{Definition}}

\newcommand{\sat}{\mathbf{Sat}}
\newcommand{\cut}{\mathbf{Co}}
\newcommand{\bd}{\mathbf{BlkDiag}}
\usepackage{amsmath,amssymb}
\usepackage{graphicx}
\usepackage{amsfonts}
\usepackage{float}
\usepackage{algorithmicx}
\usepackage{algorithm}
\usepackage{algpseudocode}
\usepackage{lipsum}
\usepackage{subfig}
\usepackage{color}
\usepackage{cite}
\usepackage{epstopdf}

\usepackage{calrsfs}
\usepackage{epstopdf}

\DeclareMathAlphabet{\pazocal}{OMS}{zplm}{m}{n}

\newcommand{\minimize}{\mbox{minimize}}
\newcommand{\tf}[1]{\mathbf{#1}}
\newcommand{\st}{\mbox{subject to}}
\newcommand{\R}{\mathbb{R}}
\newcommand{\cC}{\pazocal{C}}
\newcommand{\SFpair}{\left\{\Phix,\Phiu\right\}}
\newcommand{\Phix}{\tf \Phi_x}
\newcommand{\Phiu}{\tf \Phi_u}
\newcommand{\bPhix}{\tf{ \bar{\Phi}_x}}
\newcommand{\bPhiu}{\tf{ \bar{\Phi}_u}}
\newcommand{\Phixt}{\Phi_x}
\newcommand{\Phiut}{\Phi_u}

\algdef{SE}[DOWHILE]{Do}{doWhile}{\algorithmicdo}[1]{\algorithmicwhile\ #1}%

\title{\LARGE \bf
System Level Synthesis with State and Input Constraints
}
\author{Yuxiao Chen and James Anderson
\thanks{Yuxiao Chen is with the Department of Mechanical and Civil Engineering, Caltech,
        Pasadena, CA, 91106, USA. Emails:
        {\tt\small chenyx@caltech.edu}}
\thanks{James Anderson is with the Computing and Mathematical Sciences Department, Caltech,
         Pasadena, CA, 91106, USA. Emails:
        {\tt\small james@caltech.edu}}
}

\begin{document}
\maketitle
\thispagestyle{empty}
\pagestyle{empty}

\begin{abstract}
This paper addresses the problem of designing distributed controllers with state and input constraints in the System Level Synthesis (SLS) framework. Using robust optimization, we show how state and actuation constraints can be incorporated into the SLS structure. Moreover, we show that the dual variable associated with the constraint has the same sparsity pattern as the SLS parametrization, and therefore the computation distributes using a simple primal-dual algorithm. We provide a stability analysis for SLS design with input saturation under the Internal Model Control (IMC) framework. We show that the closed-loop system with saturation is stable if the controller has a gain less than one. In addition, a saturation compensation scheme that incorporates the saturation information is proposed which improves the naive SLS design under saturation.
\end{abstract}

\section{Introduction}
System Level Synthesis (SLS) is a recently developed framework for formulating optimal distributed control problems. Traditional methods seek to minimize the closed-loop map from exogenous disturbance to error, formally they try to solve
\begin{align*}
\underset{\tf K} {\minimize} \quad & \|f(\tf P, \tf K)\| \\
\st \quad & \tf K~ \text{stabilizes }\tf{P}, \quad \tf K \in \cC,
\end{align*}
where $f(\tf P, \tf K)\ = \tf P_{11} + \tf P_{12}\tf K(I-\tf P_{22}\tf K)^{-1}\tf P_{21}$. The distributed nature of the problem is encapsulated in the information constraint $\tf K \in \cC$. The sub-space  $\cC$ encodes for sparsity of the controller structure as well as various delays incurred through communication, actuation, and sensing. The seminal work of Rotkowitz and and Lall~\cite{RotL05} characterizes when this problem has a convex solution. In later work the conditions for convexity in~\cite{LesL11} were shown to be both necessary and sufficient.

In comparison, the SLS framework~\cite{anderson2019system} does not consider the map $\bar {\tf z} = f(\tf P,\tf K)\tf w$. Instead, it takes into account the system state and optimizes the closed-loop map from disturbance to state and disturbance to control action:
\begin{align} \label{eq:sysresp}
\begin{bmatrix} \tf x \\ \tf u \end{bmatrix} = \begin{bmatrix} \Phix \\ \Phiu \end{bmatrix}\tf w,
\end{align}
where
\begin{align*}
\Phix = (zI-A -B\tf K)^{-1},
\Phiu = \tf K(zI-A -B\tf K)^{-1}.
\end{align*}
One of the central results of SLS is the ability to design the \emph{system response} pair $\SFpair$ using convex programming (see Section \ref{sec:SLSback} for details). In this work we show how constraints on the state and control action can be incorporated into the SLS formulation. In related work~\cite{DeaTMR18}, Dean et al  consider the problem of \emph{safely learning} LQR dynamics which incorporate system constraints in an SLS setting.

\subsection{Preliminaries}
In this section we provide the necessary background material on SLS distributed control.

\subsubsection{Nomenclature} $\mathbb{R}^n$ denotes the $n$-dimensional Euclidean space. Transfer matrices and signals are denoted in bold, constant matrices and vectors are not. Matlab notation is used to denote the vertical concatenation of two objects of appropriate dimension, e.g., $[A ; B] = [A^\intercal~B^\intercal]^\intercal$.

$\pazocal{P}(P,q)=\left\{x\mid Px\le q\right\}$ denotes the polyhedron defined with $P$ and $q$. $\bd(A_1,A_2,\hdots ,A_k)$ represents the block diagonal matrix with diagonal blocks $A_1,A_2,\hdots ,A_k$.
\subsubsection{System Level Synthesis}\label{sec:SLSback}
Here we provide a very brief overview of the system level synthesis framework for distributed control. This paper is intended to be self contained, but the reader is referred to~\cite{anderson2019system, WanMD18} for  a more complete picture of both theory and computation.
We consider a linear time-invariant plant model of the form
$x(k+1) = Ax(k) + Bu(k) + w(k)$,
where $x(k), w(k) \in \R^{n}$ are the state and noise vectors at time $k$, and $u(k) \in \R^m$ is the control action at time $k$. The control synthesis problem is to design a dynamic state-feedback policy $\tf u = \tf K \tf x$. The system level synthesis approach to control synthesis is based on the notion of designing the closed-loop system responses $\SFpair$ (by choice of $\tf K$). Any stable and strictly-proper transfer matrices $\SFpair$ that satisfy the affine expression
\begin{equation}\label{eq:affine}
\begin{bmatrix} zI - A & -B\end{bmatrix} \begin{bmatrix} \Phix \\ \Phiu \end{bmatrix}=I
\end{equation}
can be used to construct an internally stabilzing controller $\tf K = \Phiu \Phix^{-1}$. In this work we find it more convenient to work in the time-domain. We thus work with a convolutional representation of the system response given by
\begin{equation*}
x(k) = \sum_{t=1}^{k+1}\Phixt[t]w(k-t), \quad u(k) =  \sum_{t=1}^{k+1}\Phiut[t]w(k-t).
\end{equation*}
The relationship between $\Phix,\Phiu$ and $\Phixt [1],\hdots, $and $\Phiut [1],\hdots$ is given through the spectral decomposition of a transfer matrix: $\Phix = \sum_{k=0}^{\infty} \Phixt[k]z^{-k}$, $\Phiu = \sum_{k=0}^{\infty} \Phiut[k]z^{-k}$. Note that~\eqref{eq:affine} imposes the constraint $\Phixt[1] = I$.

\section{SLS with State and Input Constraint}
In this section, we consider the SLS synthesis problem with state and input constraint.

The system dynamics are modeled by a discrete time linear model:
\begin{equation}\label{eq:dyn}
  x(k+1)=Ax(k)+Bu(k)+w(k),
\end{equation}
as described in Section~\ref{sec:SLSback}. We assume that a bound on exogenous disturbance $w_k$ is known and is provided as part of the problem specification. The bounds we use will be time invariant, therefore we drop time dependencies from our notation when describing these constraints. Let the disturbance constraint be modeled as $Gw\le g$, where $G\in\mathbb{R}^{q\times n}$, $g\in\mathbb{R}^q$ specify a polytopic bound on $w$.  Under this assumption, our goal is to design a distributed controller such that the state $x$ and control input $u$ are bounded for all time. To write this compactly, we use $H[x;u]\le h$ to denote the state and input bound, where $H\in\mathbb{R}^{p\times (n+m)}$, $h\in\mathbb{R}^p$. Although, the constraint is assumed to be time invariant,  the proposed method can be easily extended to the time varying formulation. The state bound may come from the specification for the design, whereast the input bound usually comes from actuator saturation. Furthermore, locality (in space and time) constraints can be easily incorporated into the synthesis procedure. Informally, spatial locality is encoded via sparsity constraints on the matrices $\Phixt[i], \Phiut[i]$ for all $i$ and a finite impulse response (FIR) constraint sets $\Phixt[t]=0, \Phiut[t]=0$ for all $t>T$. The following lemma links the different components and constraints together.
\begin{lemma}
  For a closed-loop stable FIR system with horizon $T$, if the closed-loop response is described by~\eqref{eq:sysresp}, then the following statements are equivalent:
  \begin{enumerate}
    \item For all allowable $w$, for all $t\ge0$, $H[x(t);u(t)]\le h$
    \item For all allowable $w(0:T-1)$,
    \begin{equation*}\label{eq:cons_phi}
      H\begin{bmatrix}
         \sum\nolimits_{i = 0}^{T-1} {{\Phi _x}[T - i ]} w(i) \\
         \sum\nolimits_{i = 0}^{T-1} {{\Phi _u}[T - i ]} w(i)
       \end{bmatrix}  \le h,
    \end{equation*}
    where  $w(0:T-1)$ denotes $w(0), \hdots, w(T-1)$.
  \end{enumerate}
\end{lemma}
\begin{proof}
  It is trivial to see that 2) is a necessary condition to 1) by taking $t=T$; For sufficiency, notice that the constraint on $w$, $x$ and $u$ are time invariant, and the system response is FIR with horizon $T$.
\end{proof}
Using the lemma above, we can now state the state- and input-constrained distributed SLS problem:
\begin{equation}\label{eq:original_problem}
\begin{aligned}
  \mathop {\min }\limits_{{\Phix},{\Phiu}} &J({\Phix},{\Phiu})\\
 \rm{s.t.}\; &\begin{bmatrix}
    Iz-A & -B
  \end{bmatrix}
  \begin{bmatrix}
    \Phix \\
    \Phiu
  \end{bmatrix}=I,\\
  &\Phix,\Phiu\in\frac{1}{z}\pazocal{RH}_\infty \cap \pazocal{S}\\
  &        \forall w \in \pazocal{P}(\widehat{G},\hat g),~H\begin{bmatrix}
  \sum\nolimits_{i = 0}^{T-1} {{\Phi}[T - i ]} w(i) \end{bmatrix}\le h
  \end{aligned}
\end{equation}
where $\widehat{G}=I_T\otimes G$, $\hat{g}=\mathbf{1}_T \otimes g$ with $\otimes$ denoting the Kronecker product. The set $\pazocal{S}$ encodes the FIR constraint and spatial locality constraint, see~\cite{WanMD18} for details. We refer to~\eqref{eq:original_problem} as the robust SLS problem.

The cost functional $J$ is chosen to be convex in  $\Phix$ and $\Phiu$, a typical choice would be the $\pazocal{H}_2$-norm. This is a robust convex optimization problem in the sense that the last constraint is required to hold over  all possible disturbances $w$. We use the dualization method that converts this robust optimization into a tractable convex program.
\begin{lemma}\label{lemma:Robust_optimization}
  Consider the following robust optimization problem:
  \begin{equation}\label{eq:robust_opt}
    \begin{aligned}
  \mathop {\min }\limits_\alpha\;\;&J(\alpha)  \\
  \rm{s.t.} \;&\forall \beta\in\left\{\beta\mid G\beta\le g\right\},\\
   &H_1\beta  + \alpha^\intercal H_2  \beta  + H_3\alpha  \le h, \hfill \\
\end{aligned}
  \end{equation}
  where $\alpha$ is the decision vector, $J(\cdot)$ is a convex cost function, $\beta$ is the uncertain variable, and $G \beta\le g$ is the bound for uncertainty. $H_{1,2,3}$ are constant matrices of appropriate dimensions. The robust optimization problem \eqref{eq:robust_opt} is equivalent to the following convex program:
  \begin{equation}\label{eq:dual_robost_opt}
    \begin{aligned}
  \mathop {\min }\limits_{\alpha ,\lambda } \;&J(\alpha) \; \\
  \rm{s.t.}\;&H_3\alpha  + \lambda g \le h \hfill \\
  &H_1 + \alpha^\intercal H_2 = \;\lambda G \hfill \\
  &\lambda \ge 0, \hfill \\
\end{aligned}
  \end{equation}
  where $\lambda$ is the dual variable.
\end{lemma}
\begin{proof}
  This is a specific instantiation of the general theory developed in \cite{ben2015deriving}, see the proof therein.
\end{proof}

For notational convenience let $\tf \Phi  = [\Phix; \Phiu]$. Applying Lemma \ref{lemma:Robust_optimization} to the robust SLS synthesis problem \eqref{eq:original_problem}, the robust optimization is transformed into the following convex optimization:
\begin{equation}\label{eq:conv_problem}
  \begin{aligned}
  \mathop {\min }\limits_{{\Phix},{\Phiu},\Lambda\ge0} &J({\Phix},{\Phiu})&\\
 \rm{s.t.}\; &\begin{bmatrix}
    Iz-A & -B
  \end{bmatrix}
  \begin{bmatrix}
    \Phix \\
    \Phiu
  \end{bmatrix}=I&~~~~(a)\\
  &\Phix,\Phiu\in\frac{1}{z}\pazocal{RH}_\infty \cap \pazocal{S}&~~~~(b)\\
  & H\Phi[k]  = \Lambda[k]\widehat{G},\forall k=1,...,T&~~~~(c)\\
       &\Lambda\hat{g}\le h&~~~~(d)
  \end{aligned}
\end{equation}
In this setting, the variable $\lambda$ from~\eqref{eq:dual_robost_opt} is now $T$ matrices $\Lambda[1:T],\Lambda[i]\in\mathbb{R}^{p\times q}$ and the inequality is applied in an element-wise manner.  One issue with this approach is that it requires an additional variable $\Lambda$ of dimension $\mathbb{R}^{p\times q\times T}$, which can be large. We will show that when the state space is decomposed into patches, and the state and input constraints are uncoupled, $\Lambda$ has the same sparsity pattern as the SLS parametrization $\Phi_x$ and $\Phi_u$.

Suppose the state is decomposed into $k$ patches: $x=[x_1;x_2;\hdots x_k]$, where $x_i\in\mathbb{R}^{n_i}$ are of different dimensions and $\sum_{i=1}^{k}{n_i}=n$. Accordingly, $w$ is decomposed in the same pattern $w=[w_1;w_2;\hdots w_k]$.

\begin{defi}
  The closed loop performance constraint $\pazocal{P}(H,h)$ is decoupled if $H$ is block diagonal under decomposition $x=[x_1;x_2;\hdots x_k]$, $u=[u_1;u_2;\hdots u_k]$. The bound $\pazocal{P}(G,g)$ on exogenous disturbance $w$ is decoupled if $G$ is block diagonal under decomposition $w=[w_1;w_2;\hdots w_k]$.
\end{defi}
\begin{thm}\label{thm:lambda_sparsity}
  If the performance constraint $\pazocal{P}(H,h)$ and the disturbance constraint $\pazocal{P}(G,g)$ are both decoupled, $\Lambda$ has the same sparsity pattern as $\Phi_x$ and $\Phi_u$.
\end{thm}
\begin{proof}
  By assumption, we can write
  \begin{align*}
  H &=\bd(H_1,H_2,\hdots H_k), \\  G& =\bd(G_1,G_2,\hdots G_k).
  \end{align*}
   Take $t=1$ as an example. The equality constraint for $\Phi[1]$ and $\Lambda[1]$ can be written in blocks:
  \begin{equation*}\label{eq:eq_cons}
  \resizebox{.9\hsize}{!}{$
    \begin{bmatrix}
      H_1\Phi^{1,1}[1] & \cdots & H_1\Phi^{1,k}[1] \\
      \vdots & \ddots &\vdots\\
      H_k\Phi^{k,1}[1] &\cdots & H_k\Phi^{k,k}[1] \\
    \end{bmatrix}
    =\begin{bmatrix}
      \lambda^{1,1}[1]G_1 &\cdots & \lambda^{1,k}[1]G_k \\
      \vdots  & \ddots &\vdots\\
      \lambda^{k,1}[1]G_1 &\cdots & \lambda^{k,k}[1]G_k \\
    \end{bmatrix},
    $}
  \end{equation*}
  which can be decomposed into blocks as
  \begin{equation*}\label{eq:block_eq_cons}
    \forall i,j\in\left\{1,\hdots ,k\right\},\forall t\in\left\{1,\hdots ,T\right\}, H_i\Phi^{i,j}[t] = \lambda^{i,j}[t]G_j.
  \end{equation*}

Likewise, the inequality constraint in block form is
\begin{equation}\label{eq:block_ineq_cons}
  \forall i = 1,\hdots ,k,\sum\nolimits_{t = 1}^T {\sum\nolimits_{j = 1}^k {{\lambda ^{i,j}}[t]{g_j} \le {h_i}} }.
\end{equation}
Due to the locality constraint $\{\Phix,\Phiu\}\in \pazocal{S}$, many blocks of $\mathbf{\Phi}$ will be zero. For those zero blocks $\Phi^{i,j}[t]$, the equality constraint becomes ${\lambda ^{i,j}}[t]{G_j} = 0$. Note that since $\lambda\ge0$,
\[{\lambda ^{i,j}}[t]{g_j} \ge {\lambda ^{i,j}}[t]{G_j}w = 0 \cdot w = 0\]
This means that for any feasible solution of the original optimization, we can change $\lambda^{i,j}[t]$ to zero and the solution would still be feasible. Besides, $\lambda$ does not appear in the cost function, therefore, we can simply set $\lambda ^{i,j}[t]=0$ for all $\Phi^{i,j}[t]=0$.
\end{proof}
In \cite{matni2017scalable}, the authors showed that the SLS constraints ((a),(b) in \eqref{eq:conv_problem}) are column-wise separable, therefore, when the objective function is separable, for example the $H_2$-norm, the whole synthesis can be decomposed into patches. Theorem \ref{thm:lambda_sparsity} showed that the equality constraint induced by the state and input constraint ((c) in \eqref{eq:conv_problem}) can be column-wise separated, and even inherits the sparsity pattern of $\Phi$. Unfortunately, the last inequality constraint ((d) in \eqref{eq:conv_problem}) is not column-wise separable. However, at this point, the patches are only coupled by several linear constraints and the coupling only exists between neighbors determined by the locality constraint. Therefore, one can still decentralize the synthesis in \eqref{eq:conv_problem} and resolve the coupling with a simple primal-dual algorithm. Let the dual variable associated with the inequality constraint be $\sigma\in\mathbb{R}^{1\times p}$, then the primal-dual update is as follows:
\begin{equation}\label{eq:primal_dual}
\begin{aligned}
  &\text{Primal update:}\\
  &\;\;\;\;{\begin{aligned}
  \mathop {\min }\limits_{{\Phix},{\Phiu},\Lambda\ge0} &J({\Phix},{\Phiu})+Tr[(\hat{g}\cdot \sigma)\Lambda]  \\
 \rm{s.t.}\; &\begin{bmatrix}
    Iz-A & -B
  \end{bmatrix}
  \begin{bmatrix}
    \Phix \\
    \Phiu
  \end{bmatrix}=I\\
  &\Phix,\Phiu\in\frac{1}{z}\pazocal{RH}_\infty \cap \pazocal{S}\\
  & H\begin{bmatrix}
         {\Phi}[T] & \hdots &{\Phi}[1]
       \end{bmatrix}  = \Lambda\widehat{G}
  \end{aligned}}\\
  &\text{Dual update:}\\
  &\;\;\;\;\sigma = \max(0,\sigma+ \alpha(\Lambda \hat{g}-h )),
\end{aligned}
\end{equation}
where $\alpha$ is the step size, and the stopping criteria is
\begin{equation}\label{eq:stop_criteria}
  (\Lambda\hat{g}\le h) \textbf{ and }  (\left|\sigma(\Lambda \hat{g}-h)\right|\le \epsilon),
\end{equation}
where $\epsilon$ is the tolerance. The first condition is on primal feasibility and the second condition is on complementary slackness. Note that the primal update is column-wise separable and the dual update only requires local communication between neighbors.

\begin{rem}
Compared to the invariant set method~\cite{kolmanovsky1998theory,rakovic2010parameterized,chen2018data,chen2018compositional}, that also guarantees closed-loop state bounds under a bounded disturbance signal, the SLS method proposed in this paper does not use the invariance condition of a set, but rather directly parameterize the closed-loop response. Therefore, the bound obtained is tight, as is shown in Section \ref{sec:sim}. In addition, the problem setup can be naturally extended to situations with time varying disturbance bounds. Such a setup is useful in situations when there is a nominal bound on disturbance but unexpected large disturbances may occur, but not frequently.
\end{rem}

%
%

\section{Input Saturation}
Sometimes even with design that respects the input bound, the actuator saturates under unexpected large exogenous disturbance. In this section, we discuss the SLS design under input saturation.

\subsection{Stability of saturated SLS}

%

We will show that the SLS controller resembles the internal model control (IMC) scheme discussed in \cite{doyle1987control,morari1989robust}, and via the existing result in IMC, we  prove stability of the system by saturating the internal model. Consider the dynamical system  \eqref{eq:dyn}. Suppose a stabilizing controller parameterized by $\Phix$ and $\Phiu$ is obtained via SLS synthesis. The proposed control structure is depicted in Fig. \ref{fig:block_diagram_sat}.

\begin{figure}[H]
  \centering
  \includegraphics[width=0.9\columnwidth]{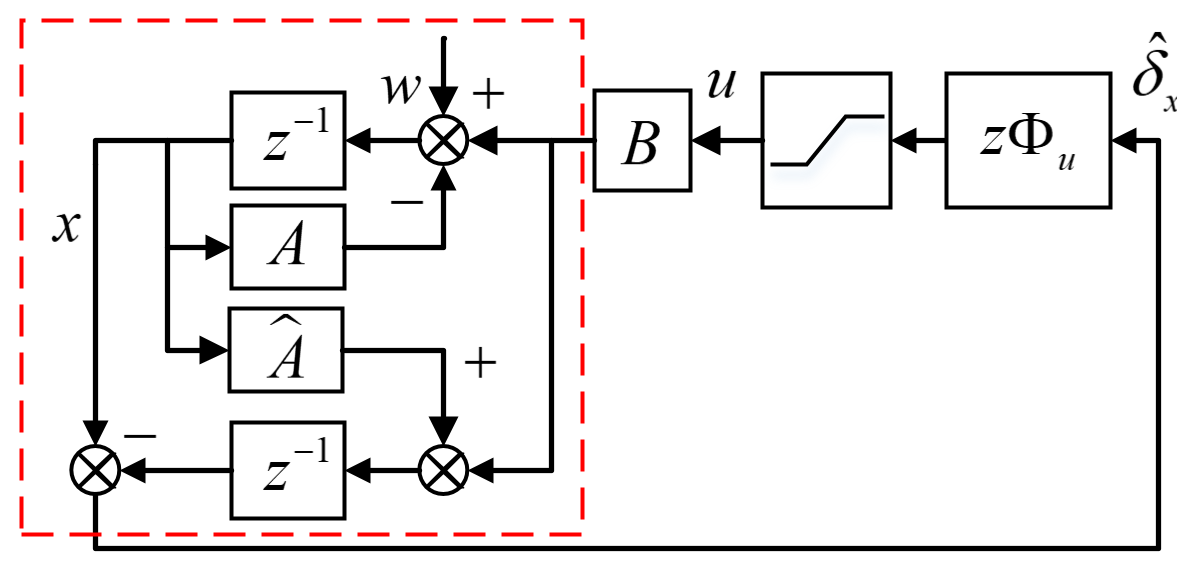}
  \caption{Block diagram with saturation.}\label{fig:block_diagram_sat}
\end{figure}

\begin{thm}\label{thm:stability}
 Assume that the dynamic system in \eqref{eq:dyn} is open-loop stable, then (i) when $\widehat A=A$ is a perfect model for the dynamics, the closed loop system depicted in Fig. \ref{fig:block_diagram_sat} is stable; (ii) when $\widehat A=A$ and the actuator is not saturated, the block diagram is a realization of the SLS controller; (iii) when $\widehat A\neq A$, but
 \begin{equation}\label{eq:imperfect_matching}
    {\left\| {(A - \hat A){{(Iz - A)}^{ - 1}}z{\Phi _u}} \right\|} < 1,
  \end{equation}
  the closed loop system is stable.
\end{thm}

\begin{proof}
  We prove stability using a small gain theorem argument. The transfer function from $Bu$ to $\hat{\delta}_x$ (the red block) in Fig. \ref{fig:block_diagram_sat} is $(A-\widehat{A})(Iz-A)^{-1}$. Clearly, when $\widehat{A}=A$ and the system is open-loop stable, the gain is zero. By constraint (b) of \eqref{eq:conv_problem}, $z\Phiu\in\pazocal{RH}_\infty$ which implies stability, therefore the overall loop gain is less than 1, and the system is stable. (iii) can be proved directly with small gain theorem.  

  For the second claim, first notice that without saturation, the closed-loop system becomes Fig. \ref{fig:block_diagram_no_sat} after some block diagram manipulation.
  \begin{figure}[H]
    \centering
    \includegraphics[width=0.9\columnwidth]{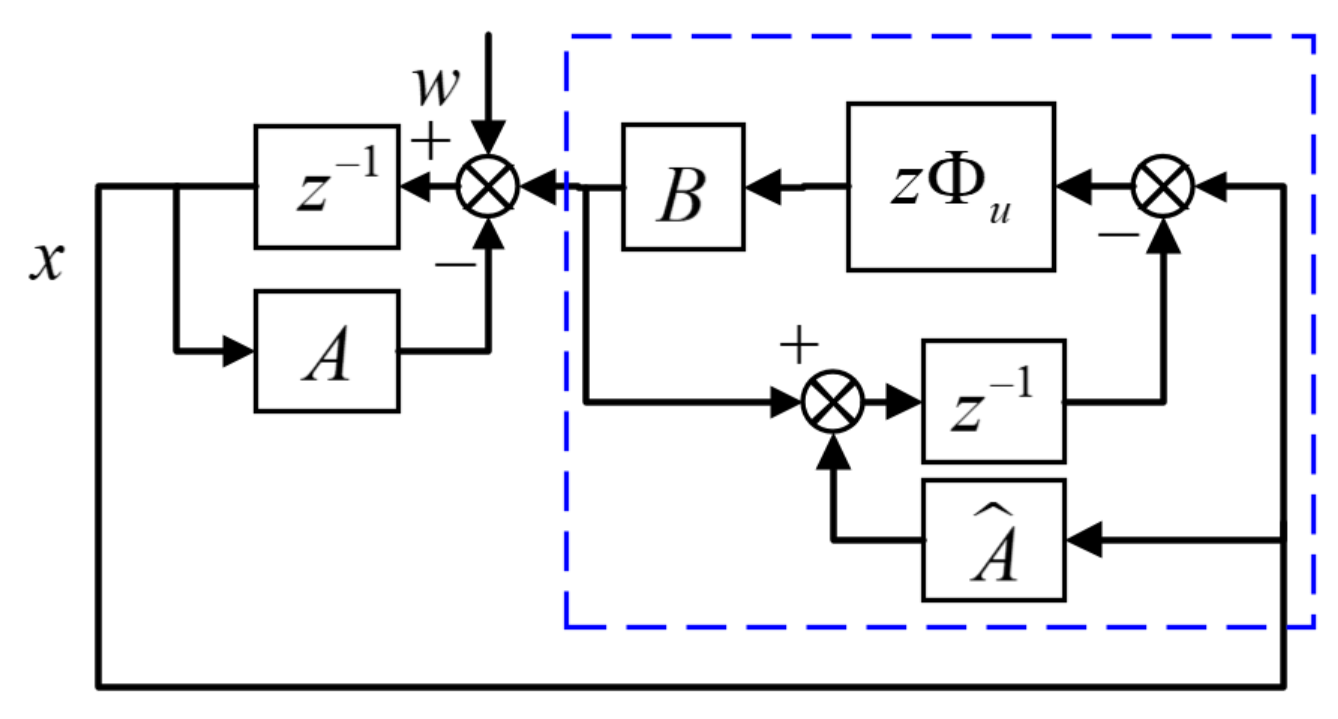}
    \caption{Block diagram without saturation.}\label{fig:block_diagram_no_sat}
  \end{figure}
  It is sufficient to show that the blocks within the blue square are equivalent to $B{\Phiu}\Phix^{ - 1}$. To see this, note that $\hat{A}=A$, and from the block diagram we have
  \begin{equation}\label{eq:block_relation}
  \begin{aligned}
    z{\Phiu}(x - {z^{ - 1}}(\widehat Ax + Bu)) &= u\\
   \Rightarrow u+\Phiu B u &=\Phiu(Iz-A)x.
    \end{aligned}
  \end{equation}
  Taking the SLS constraint (a) in \eqref{eq:conv_problem} and pre-multiplying by $(Iz-A)^{-1}$ and post-multiplying by $\Phix^{-1}$ we obtain
  \begin{equation}\label{eq:SLS_eq_derived}
    \begin{aligned}
B{\Phiu} &= (Iz - A){\Phix} - I&\;\;\;(a)\\
{(Iz - A)^{ - 1}}B{\Phiu} &= {\Phix} - {(Iz - A)^{ - 1}}&\;\;\;(b)\\
Iz - A - \Phix^{ - 1} &= B{\Phiu}\Phix^{ - 1}&\;\;\;(c)
\end{aligned}
  \end{equation}
 respectively.  Utilizing these expressions, and starting by multiplying~\eqref{eq:block_relation} on both sides by $B$, the derivation follows as
 \begin{equation*}
 Bu + B{\Phiu}Bu = B{\Phiu}(Iz - A)x,
 \end{equation*}
 apply (a) to the LHS and rearrange to get
 \begin{equation*}
 {\Phix}Bu = {(Iz - A)^{ - 1}}B{\Phiu}(Iz - A)x.
 \end{equation*}
 Applying (b) to the RHS gives
 \begin{align*}
 {\Phix}Bu &= ({\Phix} - {(Iz - A)^{ - 1}})(Iz - A)x\\
 &\Rightarrow Bu = (Iz - A - \Phix^{ - 1})x,
 \end{align*}
 and finally, applying (c) to the RHS gives
 \begin{equation*}
 Bu = B{\Phiu}\Phix^{ - 1},
 \end{equation*}
which proves that the transfer function from $x$ to $Bu$ is $B{\Phiu}\Phix^{ - 1}$.
\end{proof}
\subsection{Compensation controller for actuator saturation}\label{sec:compensation}
Under this IMC design, when input saturation occurs, it is possible that a large transient would occur if the system is close to instability.
This is because the state deviation caused by saturation is not treated with the IMC controller since it is not an exogenous disturbance and therefore gets canceled by the internal model. One approach to fixing this is to treat the saturation as an additional source of disturbance, i.e., $\bar{w}=B\Delta u$, where $\Delta u$ is the part of input that gets cut off due to saturation. Since the local controller can detect the actuator saturation, the saturation status can be broadcast to the neighbors alongside state $x$. With this additional signal   we can design a compensation controller that deals with the ``disturbance'' caused by saturation. Let $\bPhix$ and $\bPhiu$ denote the closed-loop state and input response to $\bar{w}$, then we have
\begin{equation}\label{eq:aug_response}
  \begin{aligned}
  x&=\Phix w +\bPhix \bar{w}\\
  u&=\Phiu w +\bPhiu \bar{w},
  \end{aligned}
\end{equation}
where $\bar{w}$ is directly known to the controller. When designing $\bPhix$ and $\bPhiu$, an additional constraint should be added. For the saturated node $i$, the control capacity is zero, i.e., $\bPhiu(i,:)=0$. Since the control policy is localized, the $i$-th node would have to compute $N_i$ compensation controllers that correspond to $N_i$ different saturation scenarios, where $N_i$ is the number of neighbors for node $i$ (including itself).
\begin{rem}
We only consider the case where one node in a local patch is saturated since the compensation controller with one saturated node is guaranteed to have the same locality structure as the original SLS design. However, the compensation controller with more than one saturated node may not be feasible due to communication limitations. When multiple neighbors are saturated, we simply superimpose a compensation controller for each saturated node on top of each other.
\end{rem}
\begin{rem}
With the additional input constraint caused by actuation saturation, the synthesis of the compensation controller $\bPhiu$ may be infeasible, in which case we use the robust version of SLS introduced in \cite{matni2017scalable} to obtain a controller.
\end{rem}
With $\bPhix$ and $\bPhiu$ designed, the implementation of the compensation controller is shown in Fig. \ref{fig:comp_implement}.
\begin{figure}[H]
  \centering
  \includegraphics[width=0.9\columnwidth]{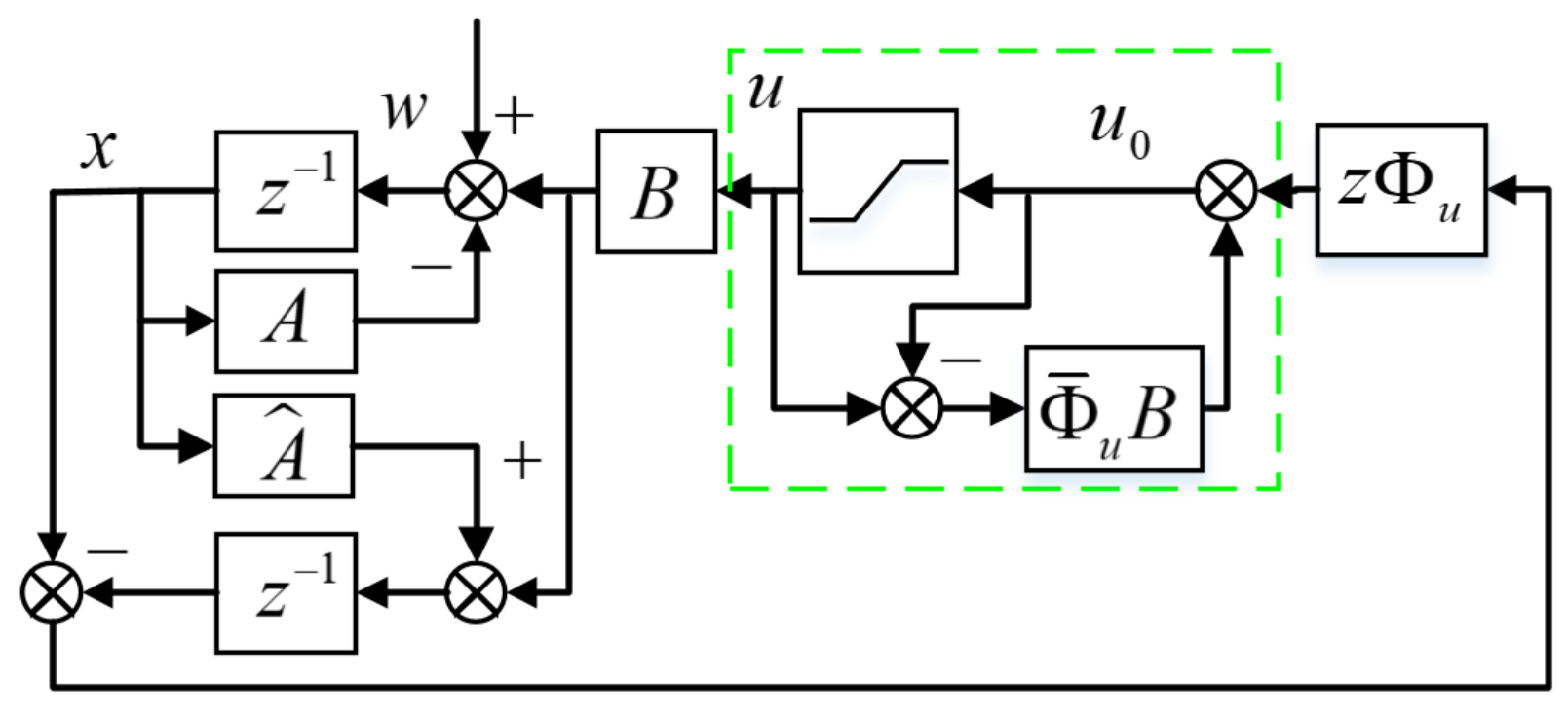}
  \caption{Closed-loop system with the compensation controller}\label{fig:comp_implement}
\end{figure}
For notational convenience, define the saturation operator as
\begin{equation*}
  \sat_{\pazocal{U}}(u)=\mathop {\arg \min }\limits_{\bar u\in\pazocal{U}} \left\| {u - \bar u} \right\|,
\end{equation*}
and the cut-off operator is defined as
\begin{equation*}
    \cut_{\pazocal{U}}(u)=u- \sat_{\pazocal{U}}(u).
  \end{equation*}
\begin{thm}
  When ${\left\| {z{\bPhiu}} \right\|}\le 1$ and the feasible input set $\pazocal{U}$ contains the origin, the block diagram in Fig. \ref{fig:comp_implement} is stable. Furthermore, when $\bPhiu = \Phiu$, the closed-loop system is equivalent to the original SLS controller with saturation.
\end{thm}
\begin{proof}
  The stability follows from small gain theorem. Based on the proof of Theorem \ref{thm:stability}, all we need to show is the stability of the blocks in the green square, which is a feedback loop consisting of the cut-off operator and $ z\bPhiu$. When $\pazocal{U}$ contains the origin, it is obvious that $\left\|\cut_{\pazocal{U}}(\cdot)\right\|<1$. Therefore, when $\left\|\bPhiu\right\|\le 1$, the compensation block is stable.

  The original SLS with saturation is shown in Fig. \ref{fig:orig_SLS_sat}.
  \begin{figure}[H]
    \centering
    \includegraphics[width=0.8\columnwidth]{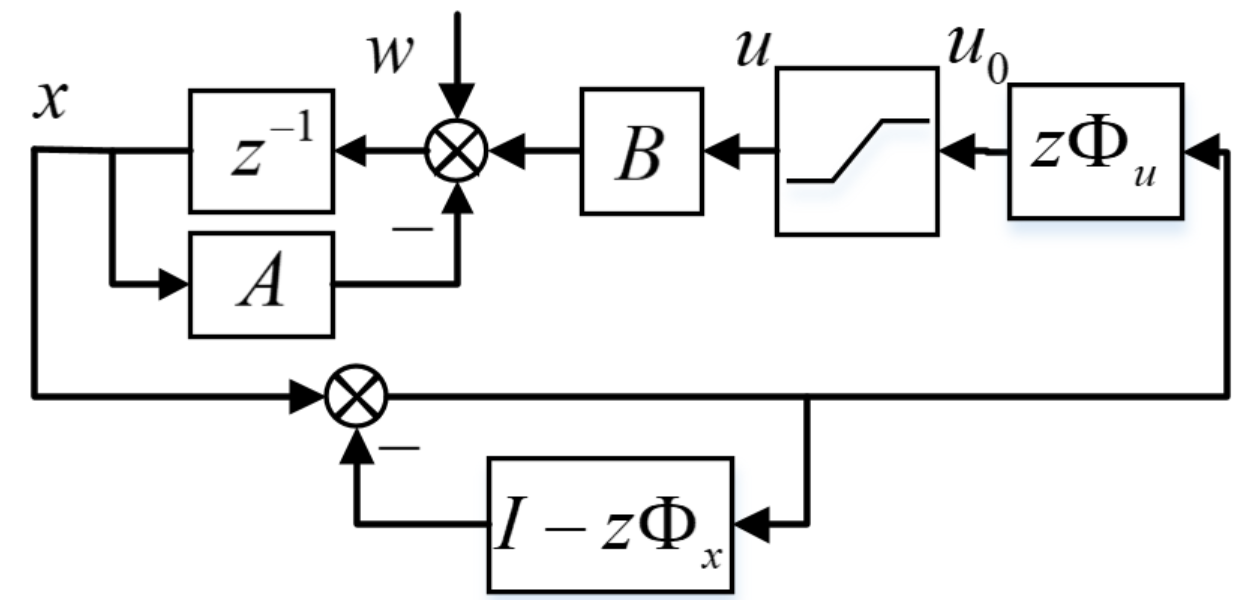}
    \caption{Closed-loop system under original SLS controller with saturation}\label{fig:orig_SLS_sat}
  \end{figure}
  The second claim can be proved by observing that when $\bPhiu=\Phiu$, the two block diagrams in Fig. \ref{fig:comp_implement} and Fig. \ref{fig:orig_SLS_sat} can be both simplified to
  \begin{equation}\label{eq:SLS_sat_relationship}
    B{\Phiu}Bu + B{\Phiu}w = (Iz - A){\Phix}B{u_0},
  \end{equation}
  where $u=\sat_{\pazocal{U}}(u_0)$.
\end{proof}
We note that the sparsity constraint $\mathbf{\bar{\Phi}}_u$ may seem restrictive, however, as we show in our numerical examples this doesn't appear to be the case. Moreover, should an infeasibility occur due to this constraint, we may apply results from robust SLS as described in~\cite{matni2017scalable} that relax the affine constraint (a) in~\ref{eq:conv_problem}.

 With $\bPhiu$ taking the control incapability caused by actuation saturation into account, we show in Section \ref{sec:sim} that it outperforms the original SLS controller.

\section{Numerical Example}\label{sec:sim}
We first show the simulation result of SLS design with state and input constraint. The model is the classic bidirectional chain system considered in \cite{matni2017scalable}. The dynamic equation is
\begin{equation}\label{eq:chain_dyn}
  x_i^ +  = \rho \left( {1 - {|\pazocal{N}_i|}\alpha } \right){x_i} + \rho\alpha \sum\nolimits_{j \in {\pazocal{N}_i}} {{x_j}}+u_i,
\end{equation}
where $x_i$ and $u_i$ are the state and control of node $i$, $\pazocal{N}_i$ denotes the set of neighbors of $x_i$ and $|\pazocal{N}_i|$ denotes the cardinality of the neighbor set. $\alpha$ and $\rho$ are parameters, we pick $\alpha=0.4$ and $\rho=1$, which makes the system neutrally stable.

 We assume an infinity norm bound on $w$, and the constraint on $x$ and $u$ is also in the infinity norm:
 \begin{equation*}
    \begin{aligned}
   \pazocal{P}(G,g)&:=\left\{ {w\mid {w_{\max }} \ge w \ge  - {w_{\max }}} \right\}\\
   \pazocal{P}(H,h)&:=\left\{ x,u\left| {\begin{array}{l}{x_{\max }} \ge x \ge  - {x_{\max }}\\
                                        {u_{\max }} \ge u \ge  - {u_{\max }}
                                        \end{array}} \right.\right\},\\
   \end{aligned}
 \end{equation*}
 where $x_{\max}\in\mathbb{R}^n$, $u_{\max}\in\mathbb{R}^m$ and $w_{\max}\in\mathbb{R}^n$ are the bound on state, input and disturbance.
 \begin{rem}
 The proposed framework naturally extends to the case with time varying bounds, for simplicity, we assume a constant bound for $w$ for all $t\ge0$.
\end{rem}

 For reasonable locality constraints and FIR horizon (horizon length $T=4$, locality radius $d=3$), the SLS synthesis is feasible. In addition, the decentralized primal-dual algorithm \eqref{eq:primal_dual} converges to the centralized synthesis solution. To validate the result, we first solve for the worst-case disturbance, i.e., $w_{0:T-1}\in\pazocal{P}(\hat{G},\hat{g})$ that is the worst case for one of the constraints, which is formulated as the following convex optimization problem:
\begin{equation}\label{eq:worst_w}
  \begin{aligned}
\mathop {\max }\limits_{w(0:T - 1)} &{H_i}\sum\nolimits_{t = 1}^T {{\Phi _x}[T - t]} w(t)\\
s.t.\;&\forall t = 0,\hdots ,T - 1,Gw(t) \le g,
\end{aligned}
\end{equation}
where $H_i$ is any row of $H$. Since \eqref{eq:worst_w} searches for the worst-case disturbance for one row of the constraint matrix, which can be either a state or input constraint, the result will depend on which constraint it tries to violate. As an example, we take the state constraint on one of the nodes in the chain system, and the worst-case disturbance is shown in Fig. \ref{fig:worst_case_w}. Because of the locality constraint, the worst-case disturbance is also local, i.e., only nonzero among neighbors of node $i$.
\begin{figure}
  \centering
  \includegraphics[width=0.8\columnwidth]{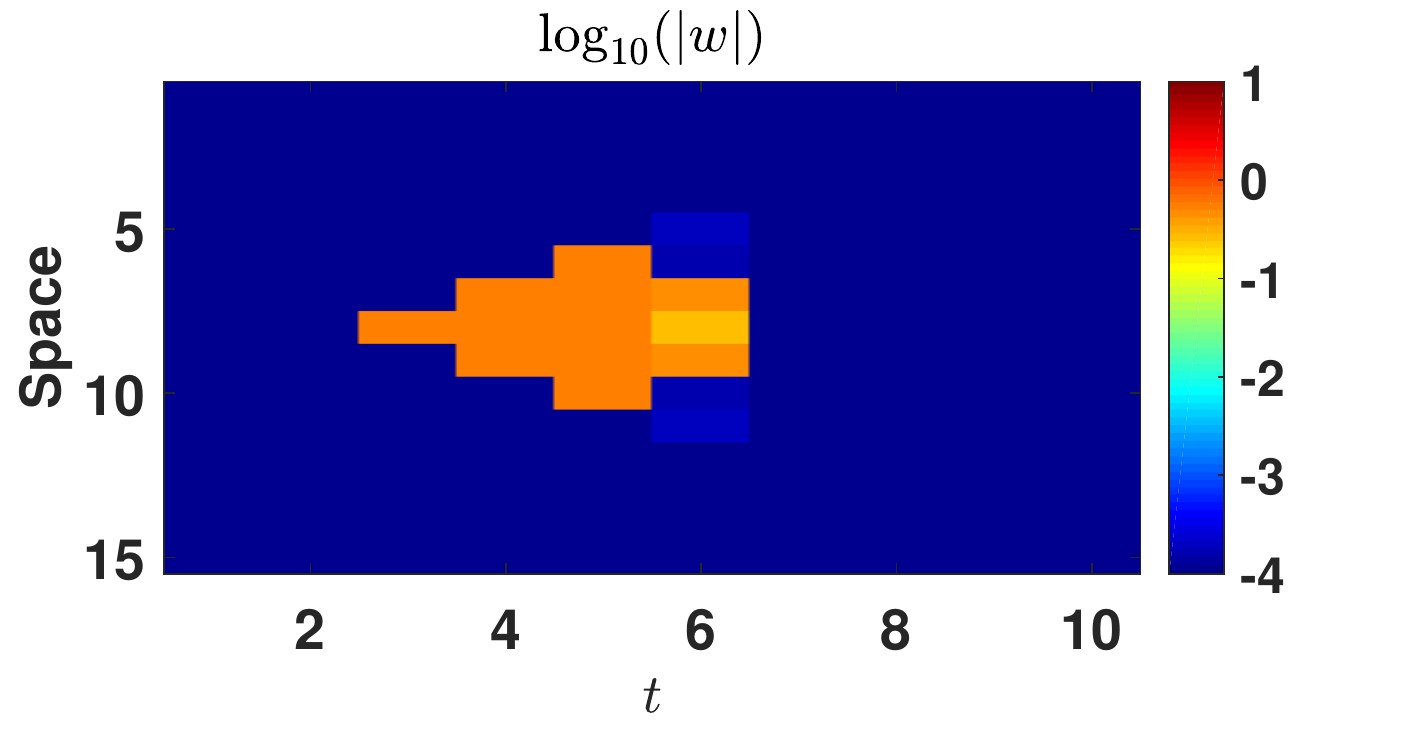}
  \caption{Worst-case disturbance.}\label{fig:worst_case_w}
\end{figure}
The closed loop response of $x_i$ is shown in Fig. \ref{fig:worst_case_xi}, and it stays within the bound specified as part of the SLS specification. Note that in general one cannot find a  worst-case disturbance  sequence  for multiple state and input constraints, therefore we simply pick one of the constraint and solve for the worst-case disturbance sequence.
\begin{figure}
  \centering
  \includegraphics[width=0.8\columnwidth]{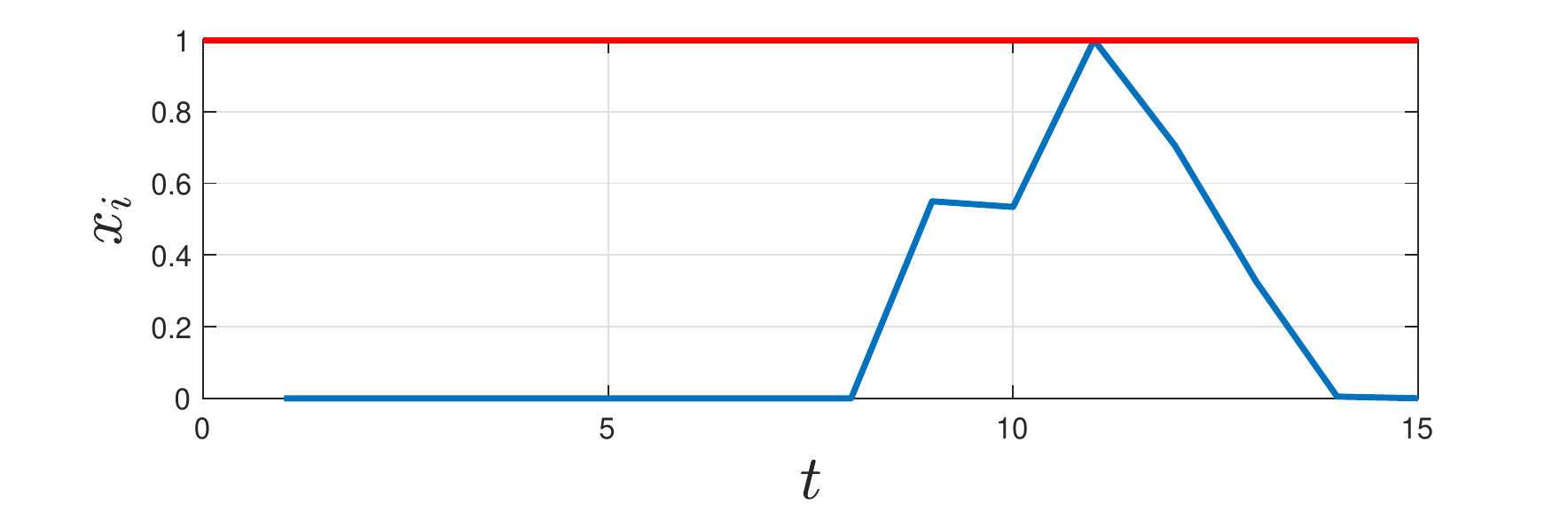}
  \caption{State response under the worst-case disturbance.}\label{fig:worst_case_xi}
\end{figure}

Next, we simulate the compensation controller. As discussed in Section \ref{sec:compensation}, the saturated IMC design may have large transient, and the IMC design with the original SLS controller $\Phiu$ as the compensation controller is equivalent to the naive SLS control with saturation. We show the comparison in Fig. \ref{fig:comp_comparison} between the naive SLS controller with saturation (top three plots) and SLS with saturation compensation, i.e., a compensation controller that assumes zero control capability for the saturated node (bottom three plots).
\begin{figure}[H]
  \centering
  \includegraphics[width=1\columnwidth]{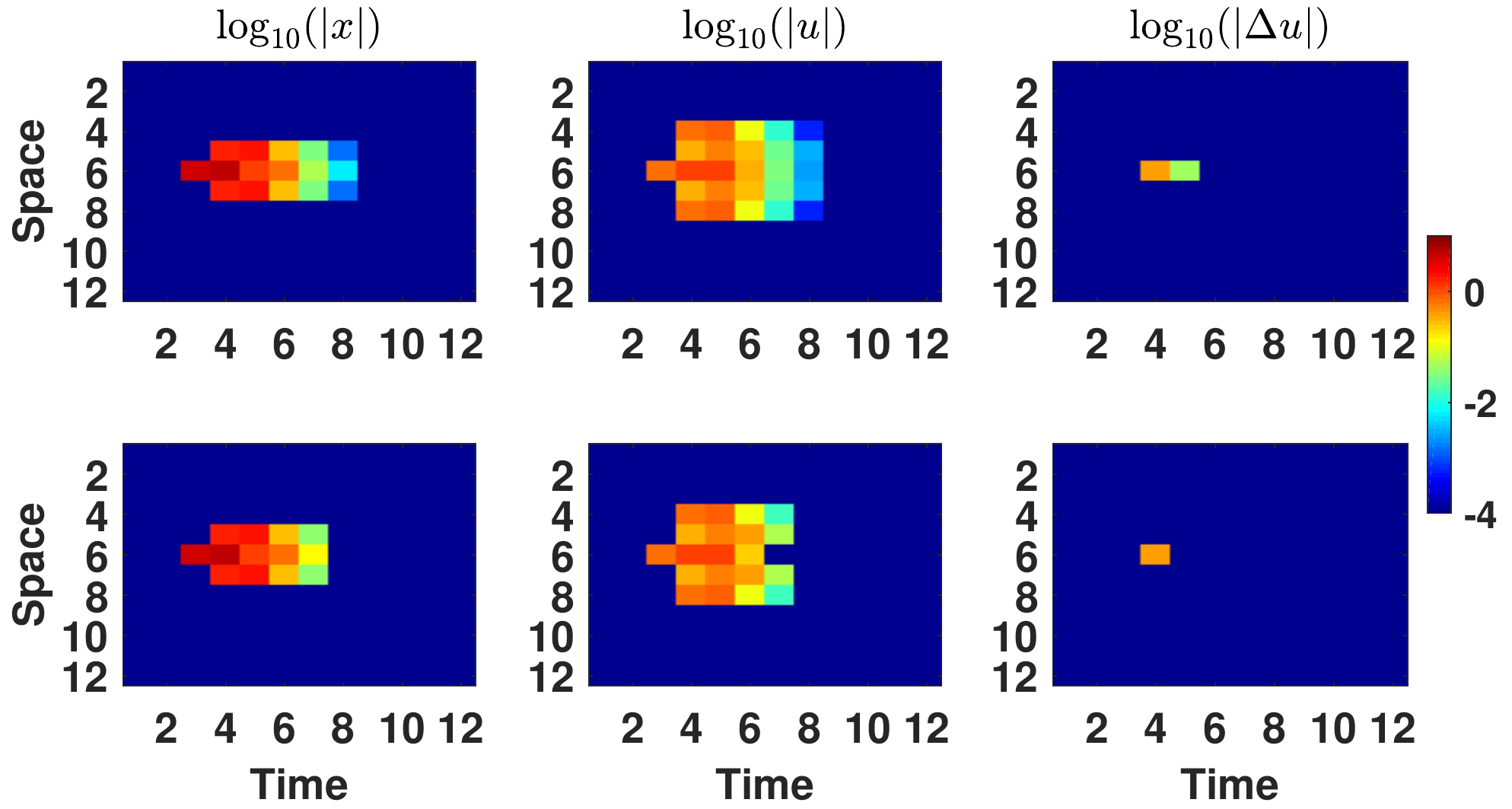}
  \caption{Simulation of naive SLS with saturation (top) and SLS with saturation compensation (bottom).}\label{fig:comp_comparison}
\end{figure}
The results show that with a compensation controller that incorporates the saturation information, the state and input converges to zero faster than the original SLS with saturation. 

\section{Conclusion}
We have described a method for incorporating state and input constraints into distributed control problems formulated using the system level synthesis framework. It was shown that we can use robust optimization to incorporate the constraints, and that the coupling constraints distribute when using a simple primal-dual algorithm. It was further shown that the closed-loop transient  dynamics can be improved upon by a slight modification to the controller (provided a simple gain condition is satisfied).

In future work we will consider the case of linear time-varying systems. The theory for SLS in this setting exists and we are currently mapping it to the constrained control setting. Such an approach will form the basis of a distributed model-predictive control scheme.


\section{Acknowledgements}
We gratefully acknowledge John C. Doyle for suggesting we look at  internal model control theory in connection with this problem.

\bibliographystyle{abbrv}
\bibliography{SLS_bib}

\end{document}